\documentclass{ceurart}
\usepackage{misc}
\usepackage{microtype}
\usepackage{xspace}
\usepackage{thmtools}

\usepackage{listings}
\newcommand{\size}[1]{\ensuremath{\lVert #1\rVert}\xspace}
\renewcommand{\sub}{\ensuremath{\mn{sub}}\xspace}

\newcommand{\tp}{\mn{tp}\xspace}
\newcommand{\roleOne}{r}
\newcommand{\roleTwo}{s}
\newcommand{\roleThree}{s'}

\newcommand{\completion}{\mathsf{cp}}
\newcommand{\separ}{\mathsf{Sep}}

\usepackage{graphbox}
\DeclareRobustCommand{\DLLogo}{%
  \begingroup\normalfont
  \kern-1.75pt\includegraphics[align=c,height=1.25\baselineskip]{dl}\kern-1.5pt%
  \endgroup
}

\usepackage{amsthm}
\newtheorem{theorem}{Theorem}
\newtheorem{lemma}[theorem]{Lemma}

\newtheorem{example}[theorem]{Example}

\begin{document}

\copyrightyear{2025}
\copyrightclause{Copyright for this paper by its authors.
  Use permitted under Creative Commons License Attribution 4.0
  International (CC BY 4.0).}

\conference{}

\title{Computation of Interpolants for Description Logic Concepts in Hard Cases}

\author[1]{Jean Christoph Jung}[%
orcid=0000-0002-4159-2255,
email=jean.jung@tu-dortmund.de
]
\address[1]{TU Dortmund University, Germany}

\author[2]{Jędrzej Kołodziejski}[%
orcid=0000-0001-5008-9224,
email=jedrzej.kolodziejski@tu-dortmund.de
]
\address[2]{TU Dortmund University, Germany}

\author[3]{Frank Wolter}[%
orcid=0000-0002-4470-606X,
email=wolter@liverpool.ac.uk
]
\address[3]{University of Liverpool, UK}

\newcommand{\nb}[1]{\textcolor{red}{!}\marginpar{\color{red}#1}}

\begin{abstract}
While the computation of Craig interpolants for description logics (DLs) with the Craig Interpolation Property (CIP)
is well understood, very little is known about the computation and size of interpolants for DLs without CIP or if one aims at interpolating concepts in a weaker DL than the DL of the input ontology and concepts. In this paper, we provide the first elementary algorithms computing (i) $\mathcal{ALC}$ interpolants between $\mathcal{ALC}$-concepts under $\mathcal{ALCH}$-ontologies and (ii) $\mathcal{ALC}$ interpolants between $\mathcal{ALCQ}$-concepts under $\mathcal{ALCQ}$-ontologies. The algorithms are based on recent decision procedures for interpolant existence.
We also observe that, in contrast, uniform depth restricted interpolants might be of non-elementary size.
\end{abstract}


\maketitle

\section{Introduction}
Interpolants between description logic (DL) concepts have found many applications. For instance, they can be used as explicit concept definitions or referring expressions, as explanations for concept inclusions, as rewritings of queries, and as separating concepts in the context of concept learning~\cite{TenEtAl06,ArtEtAl21,TomWed21,DBLP:journals/ai/JungLPW22,DBLP:conf/kr/JungLPW21}. The computation of interpolants has been investigated extensively, both by the DL community~\cite{TenEtAl13,KR2022-16,DBLP:conf/ijcai/LyonK24,DBLP:conf/aaai/KoopmannS15} but also in modal logic and related fragments of FO~\cite{DBLP:journals/lmcs/Sofronie-Stokkermans08,DBLP:journals/tocl/BenediktCB16,DBLP:journals/apal/FittingK15}. We quickly remind the reader how this is done: A Craig interpolant between $C$ and $D$ is a concept $E$ in the shared signature of $C$ and $D$
such that $\models C \sqsubseteq E$ and $\models E \sqsubseteq D$ (for simplicity we drop the ontology). A DL has the Craig Interpolation Property (CIP), if the existence of such an interpolant follows from $\models C \sqsubseteq D$. DLs such as $\mathcal{ALC}$, $\mathcal{ALCQ}$, and $\mathcal{ALCI}$ have the CIP~\cite{TenEtAl13}. Then, an interpolant $E$ can typically be extracted from a proof of $\models C \sqsubseteq D$ (or, equivalently, of non-satisfiability of $C \sqcap \neg D$) in standard calculi in the field such as tableau, the chase, sequent calculi, or type elimination~\cite{TenEtAl13,KR2022-16,DBLP:conf/ijcai/LyonK24,DBLP:journals/tocl/BenediktCB16,baldernote}.

The situation is very different for DLs that do not enjoy the CIP or if one is interested in interpolating concepts in a weaker DL than the concepts used in the inclusion. In this case, the existence of an interpolating concept does not follow from the validity of the inclusion and extracting interpolating concepts from proofs becomes much harder. In fact, very little is known about how this could be done and research has so far focused on deciding the existence of interpolants rather than constructing them~\cite{DBLP:conf/lics/JungW21,DBLP:journals/tocl/ArtaleJMOW23,DBLP:conf/kr/KuruczWZ23}. It is worth noting, however, that for extensions of $\mathcal{EL}$, the chase can be used to compute interpolants even without CIP~\cite{KR2022-16}.

It is well known that Craig interpolants of \ALC concept inclusions under \ALCH ontologies do not necessarily exist~\cite{KonevLWW09,TenEtAl13} and that not every \ALCQ concept inclusion has an \ALC interpolant (take $\models C\sqsubseteq C$ for any \ALCQ concept $C$ not equivalent to an \ALC concept). Existence of \ALC interpolants in these settings is, however, decidable~\cite{DBLP:journals/tocl/ArtaleJMOW23,kuijer2025separationdefinabilityfragmentstwovariable}. To explain the proof, assume that $\models C \sqsubseteq D$ and let $\Sigma$ be any signature (again we drop the ontology for simplicity). It is known that an interpolating $\ALC(\Sigma)$ concept exists if no
$\Sigma$-bisimilar nodes satisfying $C$ and $\neg D$ exist. Hence it suffices to decide
whether a pair of concepts is satisfiable in $\Sigma$-bisimilar nodes. 
It turns out that to decide this problem it is crucial to decide the more general problem 
whether a \emph{set} of concepts (and not just a pair) is satisfiable in mutually $\Sigma$-bisimilar nodes. By completing concepts to types containing them, it suffices to decide the latter problem for sets of types, often called \emph{mosaics}. In fact, the decision algorithms in
~\cite{DBLP:journals/tocl/ArtaleJMOW23,kuijer2025separationdefinabilityfragmentstwovariable}
use mosaics and generalize the well-known \emph{type elimination procedures} deciding satisfiability of concepts to \emph{mosaic elimination procedures} deciding $\Sigma$-bisimilar satisfiability.

Mosaic elimination procedures decide the existence of interpolants, but they do not construct any interpolants. The aim of this paper is to give the first elementary algorithms constructing $\mathcal{ALC}$ interpolants 
whenever they exist under ontologies in $\mathcal{ALCH}$ and under ontologies and concept inclusions in $\mathcal{ALCQ}$. Our algorithms are not restricted to computing Craig interpolants, but work for arbitrary signatures. The idea of the algorithms is to run the mosaic elimination procedures discussed above and construct, in addition and inductively, for each eliminated mosaic entailed $\mathcal{ALC}(\Sigma)$ concepts that witness non $\Sigma$-bisimilar satisfiability of its types. The witness concepts we propose are not aggregated at each step, but are polyadic in the sense that we define, for any set $T$ of concepts (types in the case of mosaics) which are not satisfiable in $\Sigma$-bisimilar nodes, for each $C\in T$
an $\ALC(\Sigma)$ concept $\separ(C)$ such that the following holds: 
\begin{itemize}
    \item $\models C\sqsubseteq \separ(C)$ for all $C \in T$;
    \item $\models \bigsqcap_{C\in T}\separ(C) \sqsubseteq \bot$.
\end{itemize}
The concept $\separ(C)$ constructed for $T=\{C,\neg D\}$ is then the desired interpolant. We note that an earlier attempt
to construct interpolants while running a mosaic elimination procedure without using polyadic separators does not work as stated~\cite{DBLP:journals/tocl/ArtaleJMOW23}. Hence one main contribution of this paper is to correct that proof. Our second main contribution is to show that our approach also works in the case of \ALCQ ontologies. 

%
%
%
%
%

\section{Preliminaries}

We first introduce the syntax and semantics of the basic description
logics~\ALC, \ALCH, and \ALCQ and introduce some model theory. We refer the reader to~\cite{DL-Textbook} for a comprehensive introduction to description logics.
Let \NC, and~\NR
be mutually disjoint and countably infinite sets
of \emph{concept}, and \emph{role names}. 
An \emph{\ALCQ concept} is defined according to the
syntax rule
\[
C, D ::= \top \mid A \mid \neg C \mid C \sqcap D \mid (\geq n\ r.C)
\]
where $A$ ranges over concept names, $r$ over role names, and $n\geq 0$. We use the standard abbreviations $\exists r.C$ for $(\geq 1\ r.C)$, $\forall r.C$ for $\neg\exists r.\neg C$, $C\sqcup D$ for $\neg(\neg C\sqcap \neg D)$, and $C\to D$ for $\neg C\sqcup D$.
An \emph{\ALC concept} is an \ALCQ concept in which for every subformula $(\geq 1\ r.C)$, $n$ is actually $1$.
An \emph{\ALCQ concept inclusion (\ALCQ CI)} takes the form $C \sqsubseteq D$ for \ALCQ concepts $C$ and $D$. \emph{\ALC concept inclusions} are defined accordingly. An \emph{\ALCQ ontology} is a finite set of \ALCQ CIs. An \emph{\ALCH ontology} is a finite set of \ALC concept inclusions and \emph{role inclusions (RIs)} $r\sqsubseteq s$ where $r,s$ are role names from~\NR. 
The \emph{size} of a
(finite) syntactic object $X$, denoted $\size{X}$, is the number of
symbols needed to represent it as a word, and 
the \emph{role depth} of a concept is the maximal nesting depth of concept constructors $(\geq n\ r.C)$.

As usual, the semantics is defined in terms of \emph{interpretations}
$\Imc=(\Delta^\Imc,\cdot^\Imc)$,
where $\Delta^{\mathcal{I}}$ is a non-empty set, called \emph{domain} of $\mathcal{I}$,
and $\cdot^{\mathcal{I}}$ is a function mapping
every $A \in \NC$ to a subset of $A^\Imc\subseteq \Delta^{\Imc}$ and
every $r\in\NR$ to a subset of $r^\Imc\subseteq \Delta^{\Imc} \times
\Delta^{\Imc}$.
The \emph{extension $C^{\mathcal{I}}$ of a concept
$C$ in $\mathcal{I}$} is defined as usual. 
%
An interpretation \Imc
\emph{satisfies} a CI $C \sqsubseteq D$ if $C^\Imc \subseteq
D^\Imc$ and an RI $r\sqsubseteq s$ if $r^\Imc\subseteq s^\Imc$.
We say that \Imc is a \emph{model} of
an ontology $\Omc$ if it satisfies all inclusions
in it. A concept $C$ is \emph{satisfiable under ontology \Omc} if there is a model \Imc of \Omc with $C^\Imc\neq\emptyset$. 
Moreover, \emph{$C$ is subsumed by another concept $D$ under \Omc} if $C^\Imc\subseteq D^\Imc$ in every model \Imc of $\Omc$. We write $\Omc\models C\sqsubseteq D$ in this case. 

We next introduce the studied notions and associated problems. A \emph{signature} is a set $\Sigma$ of concept and role names. An \emph{$\ALC(\Sigma)$ concept} is an \ALC concept that uses only concept and role names from $\Sigma$.  Let $\Lmc,\Lmc'$ be DLs, and let us fix an $\Lmc$ ontology \Omc, \Lmc concepts $C,D$, and a signature $\Sigma$. Then, an \emph{$\Lmc'(\Sigma)$ interpolant for $\Omc\models C\sqsubseteq D$} is an $\Lmc'(\Sigma)$ concept $E$ with 
$\Omc\models C\sqsubseteq E$ and $\Omc\models E\sqsubseteq D$. 
The associated decision problem of \emph{$\Lmc'(\Sigma)$ interpolant existence} over $\Lmc$ ontologies and concepts has been recently studied and shown decidable~\cite{DBLP:journals/tocl/ArtaleJMOW23,kuijer2025separationdefinabilityfragmentstwovariable}. The decision procedures are based on elegant model-theoretic characterizations of interpolant existence in terms of bisimulations, which we introduce next.  
A relation $Z \subseteq \Delta^{\Imc}\times \Delta^{\Jmc}$ is a
\emph{$\Sigma$-bisimulation} between interpretations $\Imc$ and $\Jmc$ if
the following conditions are satisfied for all $(d,e)\in Z$: 
\begin{description}
	\item[Atom] for all concept names $A\in \Sigma$: $d\in A^{\Imc}$ iff $e\in A^{\Jmc}$,
	\item[Back] for all role names $r\in \Sigma$ and all $(d,d')\in r^{\Imc}$, there is $(e,e')\in r^{\Jmc}$ such that $(d',e')\in Z$,

	\item[Forth] for all role names $r\in \Sigma$ and all $(e,e')\in r^{\Jmc}$, there is $(d,d')\in r^{\Imc}$ such that $(d',e')\in Z$.
\end{description}
A \emph{pointed interpretation} is a pair $\Imc,d$ with \Imc an interpretation and $d\in \Delta^{\Imc}$. We write $\Imc,d
\sim_{\mathcal{ALC},\Sigma}\Jmc,e$ and call $\Imc,d$ and $\Jmc,e$
\emph{$\Sigma$-bisimilar} if there exists an
$\Sigma$-bisimulation $Z$ such that
$(d,e)\in Z$.
%
%
We say that \ALCQ concepts $C_1,C_2$ are \emph{jointly $\sim_{\ALC,\Sigma}$-consistent under \Omc} if there are models $\Imc_1,\Imc_2$ of \Omc and elements $d_i\in C_i^{\Imc_i}$ for $i=1,2$ with $\Imc_1,d_1\sim_{\ALC,\Sigma}\Imc_2,d_2$. We have the following characterization: 

\begin{lemma}\label{lem:criterion}
  Let $\Lmc\in\{\ALCH,\ALCQ\}$, $\Omc$ be an \Lmc ontology, $C,D$ be \Lmc-concepts, and $\Sigma$ be a signature. Then the following are equivalent:
  \begin{enumerate}
  
      \item there is an $\ALC(\Sigma)$ interpolant for $\Omc\models C\sqsubseteq D$;
      \item $C,\neg D$ are not jointly $\sim_{\ALC,\Sigma}$-consistent under \Omc.
      
  \end{enumerate}
\end{lemma}
The proof of Lemma~\ref{lem:criterion} is based on the fact that $\Sigma$-bisimulations capture the expressive power of $\ALC(\Sigma)$ concepts,
and crucially relies on the use of compactness. In particular, it is not constructive in the sense that in the proof of implication $2\Rightarrow 1$, no interpolant is constructed.
We study here the associated computation problems, that is, compute the interpolants if they exist.  
A notion dual to the notion of an interpolant is that of a \emph{separator}. Given concepts $C,D,E$ we call $E$ a separator for $C,D$ if $C\sqsubseteq E$ and $E\sqsubseteq \neg D$. Clearly, $E$ is a separator for $C,D$ iff it is an interpolant for $C,\neg D$. Thus, the problems of finding interpolants and separators for a given pair of concepts are algorithmically equivalent. We will switch between these two perspectives depending on which one is more convenient in a given context.

%
%
%

\section{Role Inclusions}\label{sec:ALCH}

In this section, we are concerned with computing \ALC interpolants of concept inclusions under \ALCH ontologies. We start with an example that illustrates the failure of the computation algorithm given in~\cite{DBLP:journals/tocl/ArtaleJMOW23}.   

\begin{example}
  Fix $k\geq 1$, $\Omc=\{\roleOne\sqsubseteq\roleTwo_i,\roleTwo_i\sqsubseteq\roleThree_i\mid i\leq k\}$, $\Sigma=\{\roleThree_i, A_i\mid i\leq k\}$, and
  \begin{align*}
    C=\exists{\roleOne}.B\sqcap \forall \roleOne.(B\to \bigsqcup_{i\leq k}A_i) \hspace{1cm} \text{and} \hspace{1cm} D=\bigsqcap_{i\leq k}\forall{\roleTwo_i}.\neg A_i.
  \end{align*}
We show that $C,D$ are not jointly $\sim_{\ALC,\Sigma}$-consistent under \Omc. 
Indeed, if $C,D$ are jointly $\sim_{\ALC,\Sigma}$-consistent, then all concepts in $S=\{B\sqcap (B\to \textstyle\bigsqcup_{i\leq k}A_i)\}\cup\{\neg A_i\mid i\leq k\}$ are satisfied in mutually $\Sigma$-bisimilar nodes, which is clearly not the case. By Lemma~\ref{lem:criterion}, there is an $\ALC(\Sigma)$ interpolant for $\Omc\models C\sqsubseteq \neg D$. For instance,  $E=\bigsqcup_{i\leq k}\exists{\roleThree_i}.A_i$ is an $\ALC(\Sigma)$ interpolant. 
  The algorithm from~\cite{DBLP:journals/tocl/ArtaleJMOW23}, however, computes a concept of shape $\exists s_i'.E$ for a single $i\leq k$ and one can easily see that a concept of this shape cannot serve as an interpolant. The mistake in the algorithm is a confusion in the quantifier order in the assumptions of the interpolant construction. 
\end{example}

We first show how a natural idea for computing interpolants, which works in several other settings, fails to compute elementary sized interpolants in the presence of role inclusions. Next we provide an algorithm which does compute elementary interpolants. 

A natural idea to compute interpolants could be to show first that, if there is an $\ALC(\Sigma)$ interpolant for $\Omc\models C_0\sqsubseteq D_0$, then there is one of small role depth $n$, and then use the strongest \ALC consequence of $C_0$ of this role depth~$n$. 
Let $n\geq 0$. A \emph{$(\Sigma,n)$-uniform interpolant of $C_0$ under \Omc} is an $\ALC(\Sigma)$ concept $U$ such that $\Omc\models C_0\sqsubseteq U$ and $\Omc\models U\sqsubseteq E$ for every $\ALC(\Sigma)$ concept $E$ of role depth at most $n$ with $\Omc\models C_0\sqsubseteq E$. 
A $(\Sigma,n)$-uniform interpolant for $C_0$ under \Omc always exists, and can be used as an interpolant for $\Omc\models C_0\sqsubseteq D$ whenever an $\ALC(\Sigma)$ interpolant for $\Omc\models C_0\sqsubseteq D$ of role depth $n$ exists.
This idea has been used to compute elementary sized modal logic interpolants of $\mu$-calculus formulae~\cite{DBLP:conf/stacs/JungK25} and it follows from its proof that it applies to computing interpolants under \ALC ontologies. We actually conjecture that it works for the majority of DLs enjoying the CIP, but we leave an elaboration for future work. Unfortunately, contrary to these settings, in our case $(\Sigma,n)$-uniform interpolants need not be elementary in $n$, and consequently do not lead to elementary sized interpolants. 
In what follows we denote by \Tower the iterated exponential function, that is, $\Tower(0)=1$ and $\Tower(n+1)=2^{\Tower(n)}$.

\begin{theorem}\label{thm:uniform}
  There is an \ALCH ontology \Omc, an \ALC concept $C_0$, and signature $\Sigma$ such that there is no $(\Sigma,n)$-uniform interpolant of $C_0$ under \Omc smaller than $\Tower(n-2)$.
\end{theorem}

\begin{proof}
Consider the \ALCH ontology $\Omc=\{\roleOne\sqsubseteq\roleTwo,\roleOne\sqsubseteq\roleThree\}$,  the concept
$C_0=\exists r.\top$
and $\Sigma=\{\roleTwo,\roleThree\}$. We claim that no concept $U$ of size smaller than $\Tower(n-2)$ is a $(\Sigma,n)$-uniform interpolant for $C_0$ under \Omc. 
Assume towards contradiction that there is such an $U$. Observe that for every $C\in\ALC(\Sigma)$ of depth $n-1$, $\forall\roleTwo.C\to\exists\roleThree.C$ is an $\ALC(\Sigma)$ concept of depth $n$ and $\Omc\models C_{0} \sqsubseteq (\forall\roleTwo.C\to\exists\roleThree.C)$. Hence $\Omc\models U \sqsubseteq (\forall\roleTwo.C\to\exists\roleThree.C)$.
Consider all trees of depth $n-1$, choose one for every equivalence class of $\Sigma$-bisimulation and denote the set of all these chosen trees by \Tmc. We have $\Tower(n-1)\leq|\Tmc|$ and $\size{U}<\Tower(n-2)$. Thus, by the pigeonhole principle there are two different $\Imc_1,\Imc_2\in\Tmc$ whose respective roots $d_1,d_2$ satisfy exactly the same sub-concepts of $U$. Every two trees in \Tmc are distinguished by some $D\in\ALC(\Sigma)$ of depth $n-1$, so let us pick $D$ such that $d_1\in D^{\Imc_1}$ but $d_2\notin D^{\Imc_2}$. We claim that $\Omc\not\models U\sqsubseteq \forall\roleTwo.D\to\exists\roleThree.D$, which contradicts that $U$ is an $(\Sigma,n)$-uniform interpolant. This is witnessed by an interpretation \Jmc constructed as follows. First take the disjoint union of $\Imc_1,\Imc_2$. Then take two fresh points $e_1,e_2$, and add edges $e_1\stackrel{\roleOne}{\to}d_1,e_1\stackrel{\roleTwo}{\to}d_1,e_1\stackrel{\roleThree}{\to}d_1$ and $e_2\stackrel{\roleTwo}{\to}d_1,e_2\stackrel{\roleThree}{\to}d_2$. Since $\Omc\models C_{0}\sqsubseteq U$ and $C_0$ is true at $e_1$ we have $e_1\in U^\Jmc$. This implies $e_2\in U^\Jmc$ because $d$ and $d'$ satisfy the same subconcepts of $U$. But $\forall\roleTwo.D\to\exists\roleThree.D$ is false at $e_2$, a contradiction.

\end{proof}
On the positive side, we show the following second main result.
\begin{theorem}\label{thm:construction ALCH}
Let $\Omc$ be an \ALCH ontology, $C_0,D_0$ \ALC concepts, and $\Sigma$ be a signature. Then, if there is an $\ALC(\Sigma)$ interpolant for $\Omc\models C_0\sqsubseteq D_0$, we can construct the DAG representation of such an interpolant in time double exponential in $\size{\Omc}+\size{C_0}+\size{D_0}$.
\end{theorem}
The proof is by extending a known mosaic elimination procedure for deciding joint
$\sim_{\ALC,\Sigma}$-consistency for input ontology and concepts formulated in $\ALCH$~\cite{DBLP:journals/tocl/ArtaleJMOW23}. 
We present a slight simplification of the original procedure, as we require it only for a restricted setting. 

Let us fix an \ALCH ontology $\Omc$, $\ALC$ concepts $C_0,D_0$, and a signature $\Sigma$. We denote with $\sub(\Omc,C_0,D_0)$ the set of subconcepts that occur in \Omc, $C_0, D_0$, closed under single negation. 
A \emph{type for \Omc} is any subset of $\sub(\Omc,C_0,D_0)$ \emph{realizable} in a model of \Omc, that is, any set $t\subseteq \sub(\Omc,C_0,D_0)$ such that there is a model \Imc of $\Omc$ and element $d\in \Delta^\Imc$ with $t=\tp_\Imc(d)$ where:
\[\tp_\Imc(d)=\{E\in\sub(\Omc,C_0,D_0)\mid d\in E^\Imc\}.\]
We often treat a type $t$ as the conjunction of all concepts it contains, which allows us to write, for instance, $\Omc\models t\sqsubseteq D$. A \emph{mosaic for \Omc} is a set $T$ of types for \Omc.
We say that a type $t$ is a \emph{completion} of a concept $C\in\sub(\Omc,C_0,D_0)$ if $C\in t$, and a mosaic $T$ is a \emph{completion} of a set $\Cmc\subseteq\sub(\Omc,C_0,D_0)$ of concepts if $T$ contains a completion of every $C\in\Cmc$.

Intuitively, a mosaic $T$ describes a collection of elements in an interpretation $\Imc$ which realize precisely the types in $T$ and are mutually $\Sigma$-bisimilar. Naturally, not every set of types can be realized in this way, and we use a mosaic elimination procedure to determine which can. 
We write $t\leadsto_r t'$ if an element of type $t'$ is a viable $r$-successor of an element of type $t$, that is, $\{C\mid \forall r.C\in t\}\subseteq t'$. We will denote $\{C\mid \forall r.C\in t\} = t_{/_r}$. 
We write $T\leadsto_r T'$ if for every $t\in T$, there is $t'\in T'$ with $t\leadsto_r t'$. Let $\Mmc$ be a set of mosaics. A mosaic $T\in \Mmc$ is \emph{bad} if it violates one of the following conditions:
\begin{description}
  \item[(Atomic Consistency)]\label{it:atomic step ALCH} for every $t,t'\in T$ and $A\in\Sigma$, $A\in t$ iff $A\in t'$;
  \item[(Existential Saturation)]\label{it:exists step ALCH} for every $t\in T$ and $\exists r.C\in t$, there is $T'\in \Mmc$ such that (a) $C\in t'$ for some $t'\in T'$ with $t\leadsto_r t'$ and (b) if $\Omc\models r\sqsubseteq s$ for some $s\in \Sigma$, then $T\leadsto_s T'$.

\end{description}
Along the lines of the proof of Lemma~6.5 in~\cite{DBLP:journals/tocl/ArtaleJMOW23} one can show Lemma~\ref{lem:alc-joint-consistency} below, see the appendix for a proof sketch. The original Lemma~6.5 works with pairs of mosaics which is necessary for DLs that are not preserved under disjoint unions such as \ALCO.
\begin{restatable}{lemma}{lemalcjointconsistency}\label{lem:alc-joint-consistency}
  $C_0$ and $D_0$ are jointly $\sim_{\ALC,\Sigma}$-consistent under \Omc iff there is a set $\Mmc^*$ of mosaics that does not contain a bad mosaic and such that there is $T\in\Mmc^*$ and $t_1,t_2\in T$ with $C_0\in t_1$ and $D_0\in t_2$.
\end{restatable}

It is a consequence of Lemma~\ref{lem:alc-joint-consistency} that joint $\sim_{\ALC,\Sigma}$-consistency under \ALCH ontologies can be decided in double exponential time. Indeed, an $\Mmc^*$ as in Lemma~\ref{lem:alc-joint-consistency} can be found (if it exists) by exhaustively eliminating bad mosaics from the set of all mosaics. Since the set of all mosaics is of double exponential size, and each round of the elimination procedure can be performed in time polynomial in the size of the current set of mosaics, the upper bound follows. By the link to interpolant existence provided in Lemma~\ref{lem:criterion}, also $\ALC(\Sigma)$ interpolant existence is decidable in double exponential time.

Our aim is to extend the described mosaic elimination procedure by computing, for each type in an eliminated mosaic, its ``contribution'' to the elimination.
%
%
To formalize this we introduce a polyadic notion of a separator reflecting the fact that a mosaic may contain more than two types. Assume a set $\Cmc$ of concepts. An \emph{$\mathcal{ALC}(\Sigma)$ separator for $\Cmc$} is a function $\separ$ from \Cmc to $\mathcal{ALC}(\Sigma)$-concepts such that:
\begin{itemize}
    \item $\Omc\models C \sqsubseteq \separ(C)$ for every $C\in\Cmc$;
    \item $\Omc\models \bigsqcap_{C\in\Cmc}\separ(C) \sqsubseteq \bot$.
\end{itemize}
%
We call $\Cmc$ \emph{$\ALC(\Sigma)$-separable} if there is an $\ALC(\Sigma)$ separator for \Cmc. We will use the following lemma which connects separation of concepts with separation of their completions.

\begin{lemma}\label{lem:partial vs complete}
  Assume a set $\Cmc\subseteq\sub(\Omc,C_0,D_0)$ of concepts. The following are equivalent:
  \begin{enumerate}
    \item\label{it:sep formulae} the set $\Cmc$ is $\ALC(\Sigma)$-separable;
    \item\label{it:sep completions} every  completion $T$ of $\Cmc$ is $\ALC(\Sigma)$-separable.
\end{enumerate}
\end{lemma}

\begin{proof}
The implication \eqref{it:sep formulae}$\implies$\eqref{it:sep completions} is straightforward because a separator for $\Cmc$ is a separator for every completion $T$ of $\Cmc$.

For the other implication \eqref{it:sep completions}$\implies$\eqref{it:sep formulae} assume that for every completion of $\Cmc$ we have a separator. Let $\completion(\Cmc)$ denote the set of all functions from $\Cmc$ to types which map every $C\in\Cmc$ to one of its completions. The image $f[\Cmc]$ of every such function $f\in\completion(\Cmc)$ is a completion of \Cmc, and thus, by assumption, is separated by some $\separ_f$.

We define a separator $\separ$ for $\Cmc$ by setting, for every $C\in \Cmc$: 
\begin{equation}
\separ(C) = \bigsqcup_{\substack{t \text{ completion}\\ \text{of $C$}}}\bigsqcap_{\substack{f\in\completion(\Cmc)\\f(C)=t}}\separ_f(C)\tag{$\ast$}\label{eq:completion}
\end{equation}

To prove $\Omc\models C\sqsubseteq \separ(C)$ assume $\Imc\models \Omc$ and $d\in C^\Imc$. Let $t$ be the type $\tp_\Imc(d)$ of $d$. Clearly, $t$ is a completion of $C$. By assumption, for every $f\in\completion(\Cmc)$ and $C\in\Cmc$ we have $\Omc\models f(C)\sqsubseteq \separ_f(C)$. Hence, for every $f\in\completion(\Cmc)$ such that $f(C)=t$ we get $\Omc\models t\sqsubseteq \separ_f(C)$. It follows that $d\in(\separ(C))^\Imc$.

It remains to show that $\Omc\models\bigsqcap_{C\in\Cmc}\separ(C)\sqsubseteq\bot$. Assume towards contradiction an interpretation $\Imc\models\Omc$ with $d\in(\bigsqcap_{C\in\Cmc}\separ(C))^\Imc$. By definition of $\separ$, for every $C$ there is a completion $t_C$ such that $d$ satisfies:
\[
\bigsqcap_{\substack{f\in\completion(\Cmc)\\f(C)=t_C}}\separ_f(C).
\]
Consider a function $f\in\completion(\Cmc)$ defined as $f(C)=t_C$ for every $C$. It follows that $d$ satisfies $\separ_f(C)$ for every $C$. This contradicts the assumption that $\{\separ_f(C)\ |\ C\in\Cmc\}$ is a separator for the image of $f$ and as such is inconsistent.
\end{proof}

We inductively define separators for each eliminated mosaic. Recall that there are two ways a mosaic $T$ can be eliminated: the base case when $T$ violates atomic consistency
, and the inductive case when $T$ violates existential saturation
. We look at these cases in turn.

\smallskip\textbf{Inductive Base.} If $T$ violates atomic consistency then there is a concept name $A\in \Sigma$ and types $t,t'\in T$ with $A\in t$ and $\neg A\in t'$. Let $\separ(t)$ be $A$ if $A\in t$ and $\neg A$ otherwise.
It follows that $\Omc\models t\sqsubseteq \separ(t)$ for all $t$, and $\bigsqcap_{t\in T}\separ(t)\sqsubseteq \bot$. 

\smallskip\textbf{Inductive Step.}
Denote the current set of mosaics by $\Mmc$ and assume a mosaic $T\in\Mmc$ is eliminated because it violates existential saturation. This means that there are $t\in T$ and $\exists \roleOne.C\in t$ such that whenever $T'\in\Mmc$ satisfies (i)
$T\leadsto_\roleTwo T'$ for all $\Omc\models \roleOne\sqsubseteq\roleTwo$, and (ii) contains some $t'\in T'$ with $t\leadsto_\roleOne t'$ and $C\in t'$ then $T'\notin\Mmc$.
Consider the set:
\[
\Dmc=\{t'_{/_\roleTwo}\ |\ t'\in T,s\in \Sigma, \text{ and } \Omc\models \roleOne\sqsubseteq\roleTwo\} \cup \{\{C\}\cup t_{/_\roleOne}\}.
\]
It follows that every completion $T'$ of \Dmc was already eliminated from $\Mmc$: left and right part of the union correspond to parts (i) and (ii) of the violated condition. 
Lemma~\ref{lem:partial vs complete} provides us with a separator $\separ_\Dmc$ for $\Dmc$. We use $\separ_\Dmc$ to get a separator $\separ$ for $T$ as follows.
We put:
\begin{align*}
  \separ(t') & =\hspace{-0.15cm} \bigsqcap_{\substack{\Omc\models\roleOne\sqsubseteq\roleTwo,\\
    \roleTwo\in\Sigma}}\forall{\roleTwo}. \separ_\Dmc(t'_{/_\roleTwo})
    \hspace{0.5cm}\text{and}\hspace{0.5cm}
    \separ(t) =  \neg \bigsqcap_{t'\neq t}\separ(t')=\bigsqcup_{t'\neq t}\bigsqcup_{\substack{\\\Omc\models\roleOne\sqsubseteq\roleTwo,\\
    \roleTwo\in\Sigma}}\exists{\roleTwo}. \neg \separ_\Dmc(t'_{/_\roleTwo}). 
\end{align*}
for every $t'\neq t$.

We claim that $\separ$ separates $T$. For every $t'\neq t$ we have $\Omc\models t'\sqsubseteq \separ(t')$. This follows because for every $\Omc\models\roleOne\sqsubseteq\roleTwo$ with $\roleTwo\in\Sigma$ we have $\models t'\sqsubseteq\forall{\roleTwo}.t'_{/_\roleTwo}$ and $\Omc\models t'_{/_\roleTwo}\sqsubseteq \separ(t'_{/_\roleTwo})$.
To show $\Omc\models t\sqsubseteq \separ(t)$ assume an interpretation $\Imc\models\Omc$ with $d\in t^\Imc$. The point $d$ has an $\roleOne$-child $e$ satisfying $\{C\}\cup t_{/_\roleOne}$
and hence also $\separ_\Dmc(\{C\}\cup t_{/_\roleOne})$. By definition of a separator, the image $\separ_\Dmc[\Dmc]$ of $\separ_\Dmc$ is inconsistent. Thus, the fact that $d$ satisfies $\separ_\Dmc(\{C\}\cup t_{/_\roleOne})$ implies that some concept $E\in\separ_\Dmc[\Dmc]$ 
other than $\separ_\Dmc(\{C\}\cup t_{/_\roleOne})$ must be false at $d$.
We therefore have $d\in (\neg \separ_\Dmc(t'_{/_\roleTwo}))^\Imc$ for some $\Omc\models\roleOne\sqsubseteq\roleTwo$ with $\roleTwo\in\Sigma$ and some $t'$. Since for every $\roleTwo\in\Sigma$ with $\Omc\models\roleOne\sqsubseteq\roleTwo$ we have $d\in(t_{/_s})^\Imc$ and thus $d\in(\separ_\Dmc(t_{/_s}))^\Imc$, it follows that $t'\neq t$. This proves that $d\in(\separ(t))^\Imc$.
Note that $\separ[T]$ is inconsistent by definition: the concept $\separ(t)$ is just a negated conjunction $\bigsqcap_{t'\neq t}\separ(t')$ of the rest. This completes the proof that $\separ$ separates~$T$. 

This finishes the construction of separators for every eliminated mosaic. To construct the actual interpolant, note that Lemmas~\ref{lem:criterion} and~\ref{lem:alc-joint-consistency} imply that, if there is an $\ALC(\Sigma)$ interpolant for $\Omc\models C_0\sqsubseteq D_0$, then all completions of $\{C_0,\neg D_0\}$ have been eliminated. Lemma~\ref{lem:partial vs complete} provides us with an $\ALC(\Sigma)$ separator $\separ$ for $\{C_0,\neg D_0\}$ and it is easy to see that $\separ(C_0)$ is the sought $\ALC(\Sigma)$ interpolant.

It remains to analyze the DAG size of the constructed separators, which we do by counting the number of sub-formulae used in the constructed separators. On a high-level, we construct one formula for every type in every eliminated mosaic. This formula is of negligible size $1$ in the inductive base, so let us analyze the inductive step. This step relies on Lemma~\ref{lem:partial vs complete}, and one can see that the construction in Equation~\eqref{eq:completion} uses double exponentially many sub-formulae. It remains to note that the Lemma is invoked only double exponentially often and that the construction of the separator formulae for the just eliminated concept introduces only double exponentially many sub-formulae.
This completes the proof of Theorem~\ref{thm:construction ALCH}.

We finish the section with some remarks regarding the size of the constructed interpolants. First, we strongly conjecture that there are examples in which the interpolant is forced to have double exponential role depth, so the upper size bound in Theorem~\ref{thm:construction ALCH} is optimal. Second, it is known that the size of DAG representation of interpolants in standard DLs enjoying the CIP is at most exponential~\cite[Theorem~3.26]{TenEtAl13} and thus there is an exponential gap. 


\section{Qualified Number Restrictions}
We are concerned with computing \ALC interpolants of concept inclusions in \ALCQ under \ALCQ ontologies. We use the same notation for \ALCQ as in the previous section for \ALCH, defined in the obvious way. Our first result is that $(\Sigma,n)$-uniform \ALC interpolants can be of non-elementary size. 
\begin{theorem}\label{thm:uniform2}
  There is an \ALCQ concept $C_0$ and signature $\Sigma$ such that there is no $(\Sigma,n)$-uniform \ALC interpolant of $C_0$ smaller than $\Tower(n-2)$.
\end{theorem}
\begin{proof}
    Take the concept $C_{0}=(\leq 1\ r.\top)$, signature $\Sigma= \{r,s,s'\}$, and let $n>0$. Using $\ALC(\Sigma)$ concepts $\exists r.C \rightarrow \forall r.C$ one can show the lower bound in the same way as in the proof of Theorem~\ref{thm:uniform}.  
\end{proof}
The main result of this section is as follows.
\begin{theorem}\label{thm:constructionALCQ}
Let $\Omc$ be an \ALCQ ontology, $C_0,D_0$ \ALC concepts, and $\Sigma$ be a signature. If there is an $\ALC(\Sigma)$ interpolant for $\Omc\models C_0\sqsubseteq D_0$, then
there is one of 3-exponential size which can be constructed in 4-exponential time 
in $\size{\Omc}+\size{C_0}+\size{D_0}$.
\end{theorem}
Fix an \ALCQ ontology $\Omc$, $\ALCQ$ concepts $C_{0},D_{0}$, and a signature $\Sigma$. Our algorithm for computing interpolants again relies on a mosaic elimination procedure that determines the mosaics for which there is a model $\Imc$ of $\Omc$ which realizes the types in $T$ in mutually $\Sigma$-bisimilar nodes. To formalize the elimination condition, we need some new notation. Let $m^{\bullet}\in \mathbb{N}$ be maximal such that $(\geq m^{\bullet}\ r.C)$ occurs in $\sub(\Omc,C_{0},D_{0})$ for some $r,C$. Let $N^{\bullet}=\{0,\ldots,m^{\bullet}\}\cup \{\infty\}$, and define $<$ and $+$ on $N^{\bullet}$ as usual by setting, for instance,
$m^{\bullet}<\infty$ and $k+\infty=\infty$. For a role name $r$ and type $t$, a \emph{witnessing function} $w_{r,t}$ assigns to every type $t'$ a $w_{r,t}(t')\in N^{\bullet}$ such that for each $(\geq n\ r.C)\in \sub(\Omc,C_{0},D_{0})$,
$(\geq n\ r.C)\in t$ iff $\sum_{C\in t'} w_{r,t}(t')\geq n$. If $t$ is realizable, then there exists a witnessing function $w_{r,t}$ for each role name $r$: take a model $\Imc$ of $\Omc$ realizing $t$ in a node $d$ and define 
\begin{equation}\label{eq:eqwit}
w_{r,t}(t') =
\begin{cases}
  n & \text{ if $n=|\{ d'\in \Delta^{\Imc}\mid (d,d')\in r^{\Imc}, t'=\tp_{\Imc}(d')\}|\leq m^{\bullet}$ }\\
  \infty & \text{ otherwise.}
\end{cases}
\end{equation}
Let $T$ be a mosaic, $r$ a role name, and $(w_{r,t})_{t\in T}$ be witnessing functions. To satisfy the types in a mosaic in mutually $\Sigma$-bisimilar nodes one must be able to partition, for $r\in \Sigma$, their $r$-successors into mosaics so that the back- and-forth conditions of $\Sigma$-bisimulations hold. Our formalization of this idea follows~\cite{kuijer2025separationdefinabilityfragmentstwovariable}, but we modify the notation for our purposes. Say that a set $\mathcal{S}$ of mosaics is a \emph{mosaic partition for $(w_{r,t})_{t\in T}$} if one can assign to each $t,t'$ with $t\in T$ and $w_{r,t}(t')>0$ a non-empty set $a_{r}(t,t')\subseteq \mathcal{S}$ (intuitively, the mosaics in $\mathcal{S}$ containing $t'$ as an $r$-successor of $t$) with $t'\in T'$ for all $T'\in a_{r}(t,t')$ in such a way that 
\begin{itemize}
    \item for every $T'\in \mathcal{S}$ and $t\in T$, there exists a $t'\in T'$ with $T'\in a_{r}(t,t')$;
    \item for all types $t,t'$, $|a_{r}(t,t')|\leq w_{r,t}(t')$.
\end{itemize}
Let $\Mmc$ be a set of mosaics. A mosaic $T\in \Mmc$ is \emph{bad} if it violates one of the following conditions:
\begin{description}
  \item[(Atomic Consistency)] for every $t,t'\in T$ and $A\in\Sigma$, $A\in t$ iff $A\in t'$;
  \item[(Existential Saturation)] for every role name $r\in\Sigma$ there are witnessing functions $(w_{r,t})_{t\in T}$ and a mosaic partition $\mathcal{S}\subseteq \mathcal{M}$ for $(w_{r,t})_{t\in T}$. 
\end{description}
The following result is shown in~\cite{kuijer2025separationdefinabilityfragmentstwovariable} (using slightly different wording):
\begin{lemma}\label{lem:alcq-joint-consistency}
  (i) If the condition (Existential Saturation) is satisfied for some $T\in \Mmc$, then this is witnessed by a mosaic partition $\mathcal{S}\subseteq \Mmc$ with $|\mathcal{S}| \leq m^{\bullet}\times 2^{2|\sub(\Omc,C_{0},D_{0})|}$.
  
  (ii) $C_{0},D_{0}$ are jointly $\sim_{\ALC,\Sigma}$-consistent under \Omc iff there is a set $\Mmc^*$ of mosaics that does not contain a bad mosaic and such that there is $T\in\Mmc^*$ and $t_1,t_2\in T$ with $C_{0}\in t_1$ and $D_{0}\in t_2$.
\end{lemma}
It is a consequence of Lemma~\ref{lem:alcq-joint-consistency} that joint $\sim_{\ALC,\Sigma}$-consistency under \ALCQ ontologies can be decided in double exponential time. Indeed, an $\Mmc^*$ as in Lemma~\ref{lem:alcq-joint-consistency} can be found (if it exists) by exhaustively eliminating bad mosaics from the set of all mosaics. Since the set of all mosaics is of double exponential size, and each round of the elimination procedure can be performed in  
double exponential time, the upper bound follows. By the link to interpolant existence provided in Lemma~\ref{lem:criterion}, also $\ALC(\Sigma)$ interpolant existence is decidable in double exponential time.

We next exploit the elimination procedure to construct interpolants. 
Similar to the previous Section~\ref{sec:ALCH} we compute an $\ALC(\Sigma)$ separator for every eliminated mosaic. It will be convenient to actually compute something slightly stronger.
Let $\mathcal{M}$ be a set of mosaics. A function $\separ$ that maps 
every $t$ in some $T\in \mathcal{M}$ to an $\mathcal{ALC}(\Sigma)$ concept $\separ(t)$ 
is called \emph{general $\mathcal{ALC}(\Sigma)$ separator for $\mathcal{M}$} if for every $T\in \mathcal{M}$ the restriction of $\separ$ to $T$ is an $\mathcal{ALC}(\Sigma)$ separator for $T$. 

\medskip

We compute, by induction, a general $\mathcal{ALC}(\Sigma)$ separator for the set of eliminated mosaics.

\smallskip\textbf{Inductive Base.} Assume $T$ has been eliminated because atomic consistency is violated. Then there exists $A\in \Sigma$ such that the function $\separ_{T}$ defined by setting $\separ_{T}(t)=A$ if $A\in T$ and $\separ_{T}(t)=\neg A$ otherwise, is an $\mathcal{ALC}(\Sigma)$ separator for $T$. Let $\mathcal{E}_{0}$ be the set of all mosaics that violate atomic consistency and let $\separ_{T}$ be defined as above for $T\in \mathcal{E}_{0}$.
Then we obtain a general separator $\separ_{0}$ for $\mathcal{E}_{0}$ by setting $\separ_{0}(t) =\bigsqcap_{T\in \mathcal{E}_{0}}\separ_{T}(t)$, for all $t\in T\in \mathcal{E}_{0}$.

\smallskip\textbf{Inductive Step.} Assume $\mathcal{E}_{n}$ is the set of eliminated mosaics and $\separ_{n}$ is a general separator for $\mathcal{E}_{n}$. Let $\mathcal{M}_{n}$ be the set of mosaics that have not yet been eliminated. 
Setting $\separ_{n}(t)=\top$ for types $t$ that do not occur in any mosaic in $\mathcal{E}_{n}$, we may assume that $\separ_{n}$ is defined for all types.
Let $T$ be the mosaic eliminated in the next step. Then existential saturation is violated in $\mathcal{M}_{n}$. Note that this implies that one can pick a role name $r\in \Sigma$ such that for every witnessing functions $w_{r,t}$, $t\in T$, and mosaic partition $\mathcal{S}$ for $(w_{r,t})_{t\in T}$ there is an eliminated mosaic $T'\in \mathcal{S}\cap \mathcal{E}_{n}$. We fix such an $r$. 

Let $\Tmc$ be the set of all types. Denote by $\mathcal{V}^{+}$ the set of all conjunctions $C=\bigsqcap_{t\in \Tmc}L_{t}$ with $L_{t}\in \{\separ_{n}(t),\neg\separ_{n}(t)\}$.
For any nonempty subset $\mathcal{B}\subseteq \mathcal{V}^{+}$ we set as usual
 $$
 \nabla_{r}(\Bmc)= (\bigsqcap_{C\in \Bmc}\exists r. C) \sqcap \forall r.(\bigsqcup_{C\in \Bmc}C).
 $$
 Let $\delta_{r}(t)$ be the disjunction of all $\nabla_{r}(\Bmc)$ such that $t \sqcap \nabla_{r}(\Bmc)$ is satisfiable under $\Omc$. Observe that $\Omc\models t \sqsubseteq \delta_{r}(t)$. Take any $t_{0}\in T$ and set $\separ(t_{0})=\neg \bigsqcap_{t\in T\setminus\{t_{0}\}}\delta_{r}(t)$ and $\separ(t)=\delta_{r}(t)$ for all $t\in T\setminus\{t_{0}\}$.
 \begin{lemma}\label{lem:sepalcq}
     $\separ$ is an $\mathcal{ALC}(\Sigma)$ separator for $T$.
 \end{lemma}
 \begin{proof}
 It suffices to show that $\Omc \models t_{0} \sqsubseteq \separ(t_{0})$, the remaining conditions are trivial. Assume this is not the case. Then we have a model $\Imc_{t_{0}}$ of $\Omc$ and some $d_{t_{0}}\in \Delta^{\Imc_{t_{0}}}$ such that $d_{t_{0}}\in (t_{0} \sqcap (\bigsqcap_{t\in T\setminus \{t_{0}\}}\delta_{r}(t))^{\Imc_{t_{0}}}$.
 Pick the (unique) $\Bmc\subseteq \mathcal{V}^{+}$ such that $d_{t_{0}}\in \nabla_r(\Bmc)^{\Imc_{t_{0}}}$. Then $\nabla_r(\Bmc)$ is a disjunct
 of each $\delta_{r}(t)$ with $t\in T$. Take interpretations $\Imc_{t}$ and $d_{t}$ with $d_{t}\in (t\sqcap \nabla_r(\Bmc))^{\Imc_{t}}$ for $t\in T\setminus \{t_{0}\}$. We may assume that all $\Imc_{t}$, $t\in T$, coincide (otherwise take their disjoint union) and denote it by $\Imc$.
 Next define, for $t\in T$, $w_{r,t}(t')$ as in \eqref{eq:eqwit}.
Then $(w_{r,t})_{t\in T}$ is a family of witnessing functions. Let for each $C\in \Bmc$:
\[
T_{C} = \{ \tp_\Imc(d) \mid \text{ there is $t\in T$ with $(d_{t},d)\in r^{\Imc}$ and $d\in C^{\Imc}$}\}
\]
By the definition of \Bmc, $T_{C}\not=\emptyset$ for every $C\in\Bmc$. Observe that none of the mosaics $T_{C}$ is in $\mathcal{E}_{n}$ because otherwise $\Omc\models \bigsqcap_{t\in T_{C}}\separ_{n}(t)\sqsubseteq \bot$ which is not the case
since $\separ_{n}(t)$ is a conjunct of $C$ for all $t\in T_{C}$. 

We show that $\{T_{C}\mid C\in \Bmc\}$ form a mosaic partition for $(w_{r,t})_{t\in T}$, and so derive a contradiction. To this end define $a_{r}(t,t')\subseteq \{T_{C} \mid C \in \Bmc\}$ as
\[
\{T_{C} \mid \text{ there is $d\in C^{\Imc}$ with $(d_{t},d)\in r^{\Imc}$
and $\tp_{\Imc}(d)=t'$}\}
\]
To see that $a_{r}(t,t')$ is as required, first observe that $t'\in T_{C}$ for any $T_{C}\in a_{r}(t,t')$. Next assume that a $T_{C}$ with $C\in \Bmc$ and $t\in T$ are given.
From $d_{t}\in (t\sqcap \nabla_r(\Bmc))^{\Imc}$ we obtain a $d$ with $(d_{t},d)\in r^{\Imc}$ and $d\in C^{\Imc}$. Let $t'=\tp_{\Imc}(d)$. Then $T_{C}\in a_{r}(t,t')$. The condition $|a_{r}(t,t')|\leq w_{r.t}(t')$ follows directly from the definitions.
\end{proof}
Using Lemma~\ref{lem:sepalcq} we obtain a general $\ALC(\Sigma)$ separator $\separ_{n+1}$ for $\mathcal{E}_{n}\cup\{T\}$ by setting $\separ_{n+1}(t)=\separ(t) \sqcap \separ_{n}(t)$ for all $t\in T$
and $\separ_{n+1}(t)=\separ_{n}(t)$ for all remaining types. 

\medskip

This finishes the construction of separators for every eliminated mosaic. One can now construct
the actual interpolants in exactly the same way as in the previous section for \ALCH via Lemma~\ref{lem:partial vs complete}. To compute the size of the interpolants, observe that in the construction above $\size{\separ_{n+1}(t)}\leq \size{\separ_{n}(t)} \times 2^{2^{2^{f(m)}}}$ with $f$ a polynomial function and $m=\size{\Omc} + \size{C_{0}} + \size {D_{0}}$. As $\separ_{n}$ stabilizes
after at most double exponentially many elimination steps, $\size{\separ_{n}}$ is bound by a 
3-exponential function in $\size{\Omc} + \size{C_{0}} + \size {D_{0}}$. This bound remains 3-exponential under DAG representation. The construction of $\separ_{n}(t)$ involves satisfiability checks for concepts of 3-exponential size, so overall the interpolant can be constructed in 4-exponential time.
\section{Conclusion}

We have presented first non-trivial algorithms for computing \ALC interpolants under \ALCH and \ALCQ ontologies, relying on the new notion of polyadic separators tailored to store witnesses for the fact that a mosaic (or set of types) cannot be realized in mutually bisimilar models. Theorems~\ref{thm:construction ALCH} and~\ref{thm:constructionALCQ} demonstrate the inherent difficulty of the problem and explain why previously known techniques do not easily apply in the cases of \ALCH and \ALCQ. It is worth to note that Theorem~\ref{thm:construction ALCH} can be easily modified to obtain non-elementary lower bounds for the size of uniform interpolants at the concept level in the presence of \ALCH ontologies. These lower bounds, in turn, translate to non-elementary lower bounds for the size of uniform interpolants at the ontology level in \ALCH. This implies that the resolution based calculus for computing uniform interpolants of \ALCH ontologies from~\cite{Koopmann2014CountAF} cannot run in elementary time, answering a question posed by the authors. 

In the future, we would like to extend our algorithms to other standard DL constructors. While we believe that this will be rather easy for some constructors, like inverse roles or the universal role, we expect it to be much more involved in other cases such as nominals. In fact, already unifying the two presented algorithms into one for computing \ALC interpolants under \ALCHQ ontologies appears to be challenging.
It would be also interesting to analyze our procedures in the ontology-free cases (or with an ontology containing only role inclusions), for which we expect smaller interpolants. 
From a practical perspective, it would be interesting to extend implemented tableaux algorithms to be able to compute interpolants. Beyond description logics, it would be very interesting to compute interpolants in the guarded and/or the two-variable fragments of first-order logic~\cite{DBLP:conf/lics/JungW21} or in first-order modal logics~\cite{DBLP:conf/kr/KuruczWZ23}.  
%
%
\newpage

\bibliography{main}

\begin{thebibliography}{21}
\expandafter\ifx\csname natexlab\endcsname\relax\def\natexlab#1{#1}\fi
\providecommand{\url}[1]{\texttt{#1}}
\providecommand{\href}[2]{#2}
\providecommand{\path}[1]{#1}
\providecommand{\DOIprefix}{doi:}
\providecommand{\ArXivprefix}{arXiv:}
\providecommand{\URLprefix}{URL: }
\providecommand{\Pubmedprefix}{pmid:}
\providecommand{\doi}[1]{\href{http://dx.doi.org/#1}{\path{#1}}}
\providecommand{\Pubmed}[1]{\href{pmid:#1}{\path{#1}}}
\providecommand{\bibinfo}[2]{#2}
\ifx\xfnm\relax \def\xfnm[#1]{\unskip,\space#1}\fi
\bibitem[{ten Cate et~al.(2006)ten Cate, Conradie, Marx, and Venema}]{TenEtAl06}
\bibinfo{author}{B.~ten Cate}, \bibinfo{author}{W.~Conradie}, \bibinfo{author}{M.~Marx}, \bibinfo{author}{Y.~Venema},
\newblock \bibinfo{title}{Definitorially complete description logics},
\newblock in: \bibinfo{booktitle}{Proceedings of the 10th International Conference on Principles of Knowledge Representation and Reasoning, {KR} 2006}, \bibinfo{publisher}{{AAAI} Press}, \bibinfo{year}{2006}, pp. \bibinfo{pages}{79--89}. \URLprefix \url{http://www.aaai.org/Library/KR/2006/kr06-011.php}.
\bibitem[{Artale et~al.(2021)Artale, Mazzullo, Ozaki, and Wolter}]{ArtEtAl21}
\bibinfo{author}{A.~Artale}, \bibinfo{author}{A.~Mazzullo}, \bibinfo{author}{A.~Ozaki}, \bibinfo{author}{F.~Wolter},
\newblock \bibinfo{title}{On free description logics with definite descriptions},
\newblock in: \bibinfo{booktitle}{Proceedings of the 18th International Conference on Principles of Knowledge Representation and Reasoning, {KR} 2021}, \bibinfo{year}{2021}, pp. \bibinfo{pages}{63--73}. \URLprefix \url{https://doi.org/10.24963/kr.2021/7}. \DOIprefix\doi{10.24963/kr.2021/7}.
\bibitem[{Toman and Weddell(2021)}]{TomWed21}
\bibinfo{author}{D.~Toman}, \bibinfo{author}{G.~E. Weddell},
\newblock \bibinfo{title}{{{FO} Rewritability for {OMQ} using Beth Definability and Interpolation}},
\newblock in: \bibinfo{booktitle}{Proceedings of the 34th International Workshop on Description Logics, {DL} 2021}, \bibinfo{publisher}{CEUR-WS.org}, \bibinfo{year}{2021}. \URLprefix \url{http://ceur-ws.org/Vol-2954/paper-29.pdf}.
\bibitem[{Jung et~al.(2022)Jung, Lutz, Pulcini, and Wolter}]{DBLP:journals/ai/JungLPW22}
\bibinfo{author}{J.~C. Jung}, \bibinfo{author}{C.~Lutz}, \bibinfo{author}{H.~Pulcini}, \bibinfo{author}{F.~Wolter},
\newblock \bibinfo{title}{Logical separability of labeled data examples under ontologies},
\newblock \bibinfo{journal}{Artif. Intell.} \bibinfo{volume}{313} (\bibinfo{year}{2022}) \bibinfo{pages}{103785}. \URLprefix \url{https://doi.org/10.1016/j.artint.2022.103785}. \DOIprefix\doi{10.1016/j.artint.2022.103785}.
\bibitem[{Jung et~al.(2021)Jung, Lutz, Pulcini, and Wolter}]{DBLP:conf/kr/JungLPW21}
\bibinfo{author}{J.~C. Jung}, \bibinfo{author}{C.~Lutz}, \bibinfo{author}{H.~Pulcini}, \bibinfo{author}{F.~Wolter},
\newblock \bibinfo{title}{Separating data examples by description logic concepts with restricted signatures},
\newblock in: \bibinfo{editor}{M.~Bienvenu}, \bibinfo{editor}{G.~Lakemeyer}, \bibinfo{editor}{E.~Erdem} (Eds.), \bibinfo{booktitle}{Proceedings of the 18th International Conference on Principles of Knowledge Representation and Reasoning, {KR} 2021, Online event, November 3-12, 2021}, \bibinfo{year}{2021}, pp. \bibinfo{pages}{390--399}. \URLprefix \url{https://doi.org/10.24963/kr.2021/37}. \DOIprefix\doi{10.24963/KR.2021/37}.
\bibitem[{ten Cate et~al.(2013)ten Cate, Franconi, and Seylan}]{TenEtAl13}
\bibinfo{author}{B.~ten Cate}, \bibinfo{author}{E.~Franconi}, \bibinfo{author}{I.~Seylan},
\newblock \bibinfo{title}{Beth definability in expressive description logics},
\newblock \bibinfo{journal}{J. Artif. Intell. Res.} \bibinfo{volume}{48} (\bibinfo{year}{2013}) \bibinfo{pages}{347--414}.
\bibitem[{Fortin et~al.(2022)Fortin, Konev, and Wolter}]{KR2022-16}
\bibinfo{author}{M.~Fortin}, \bibinfo{author}{B.~Konev}, \bibinfo{author}{F.~Wolter},
\newblock \bibinfo{title}{{Interpolants and Explicit Definitions in Extensions of the Description Logic EL}},
\newblock in: \bibinfo{booktitle}{{Proceedings of the 19th International Conference on Principles of Knowledge Representation and Reasoning}}, \bibinfo{year}{2022}, pp. \bibinfo{pages}{152--162}. \URLprefix \url{https://doi.org/10.24963/kr.2022/16}. \DOIprefix\doi{10.24963/kr.2022/16}.
\bibitem[{Lyon and Karge(2024)}]{DBLP:conf/ijcai/LyonK24}
\bibinfo{author}{T.~S. Lyon}, \bibinfo{author}{J.~Karge},
\newblock \bibinfo{title}{Constructive interpolation and concept-based beth definability for description logics via sequents},
\newblock in: \bibinfo{booktitle}{Proceedings of the Thirty-Third International Joint Conference on Artificial Intelligence, {IJCAI-24}}, \bibinfo{publisher}{International Joint Conferences on Artificial Intelligence Organization}, \bibinfo{year}{2024}, pp. \bibinfo{pages}{3484--3492}. \URLprefix \url{https://doi.org/10.24963/ijcai.2024/386}. \DOIprefix\doi{10.24963/ijcai.2024/386}, \bibinfo{note}{main Track}.
\bibitem[{Koopmann and Schmidt(2015)}]{DBLP:conf/aaai/KoopmannS15}
\bibinfo{author}{P.~Koopmann}, \bibinfo{author}{R.~A. Schmidt},
\newblock \bibinfo{title}{Uniform interpolation and forgetting for {ALC} ontologies with aboxes},
\newblock in: \bibinfo{editor}{B.~Bonet}, \bibinfo{editor}{S.~Koenig} (Eds.), \bibinfo{booktitle}{Proceedings of the Twenty-Ninth {AAAI} Conference on Artificial Intelligence, January 25-30, 2015, Austin, Texas, {USA}}, \bibinfo{publisher}{{AAAI} Press}, \bibinfo{year}{2015}, pp. \bibinfo{pages}{175--181}. \DOIprefix\doi{10.1609/AAAI.V29I1.9206}.
\bibitem[{Sofronie{-}Stokkermans(2008)}]{DBLP:journals/lmcs/Sofronie-Stokkermans08}
\bibinfo{author}{V.~Sofronie{-}Stokkermans},
\newblock \bibinfo{title}{Interpolation in local theory extensions},
\newblock \bibinfo{journal}{Log. Methods Comput. Sci.} \bibinfo{volume}{4} (\bibinfo{year}{2008}). \URLprefix \url{https://doi.org/10.2168/LMCS-4(4:1)2008}. \DOIprefix\doi{10.2168/LMCS-4(4:1)2008}.
\bibitem[{Benedikt et~al.(2016)Benedikt, ten Cate, and {Vanden Boom}}]{DBLP:journals/tocl/BenediktCB16}
\bibinfo{author}{M.~Benedikt}, \bibinfo{author}{B.~ten Cate}, \bibinfo{author}{M.~{Vanden Boom}},
\newblock \bibinfo{title}{Effective interpolation and preservation in guarded logics},
\newblock \bibinfo{journal}{{ACM} Trans. Comput. Log.} \bibinfo{volume}{17} (\bibinfo{year}{2016}) \bibinfo{pages}{8:1--8:46}.
\bibitem[{Fitting and Kuznets(2015)}]{DBLP:journals/apal/FittingK15}
\bibinfo{author}{M.~Fitting}, \bibinfo{author}{R.~Kuznets},
\newblock \bibinfo{title}{Modal interpolation via nested sequents},
\newblock \bibinfo{journal}{Ann. Pure Appl. Log.} \bibinfo{volume}{166} (\bibinfo{year}{2015}) \bibinfo{pages}{274--305}. \URLprefix \url{https://doi.org/10.1016/j.apal.2014.11.002}. \DOIprefix\doi{10.1016/J.APAL.2014.11.002}.
\bibitem[{ten Cate(2022)}]{baldernote}
\bibinfo{author}{B.~ten Cate}, \bibinfo{title}{Lyndon Interpolation for Modal Logic via Type Elimination Sequences}, \bibinfo{type}{Technical Report}, ILLC, Amsterdam, \bibinfo{year}{2022}.
\bibitem[{Jung and Wolter(2021)}]{DBLP:conf/lics/JungW21}
\bibinfo{author}{J.~C. Jung}, \bibinfo{author}{F.~Wolter},
\newblock \bibinfo{title}{Living without beth and craig: Definitions and interpolants in the guarded and two-variable fragments},
\newblock in: \bibinfo{booktitle}{Proceedings of the 36th Annual {ACM/IEEE} Symposium on Logic in Computer Science, {LICS} 2021}, \bibinfo{publisher}{{IEEE}}, \bibinfo{year}{2021}, pp. \bibinfo{pages}{1--14}. \URLprefix \url{https://doi.org/10.1109/LICS52264.2021.9470585}. \DOIprefix\doi{10.1109/LICS52264.2021.9470585}.
\bibitem[{Artale et~al.(2023)Artale, Jung, Mazzullo, Ozaki, and Wolter}]{DBLP:journals/tocl/ArtaleJMOW23}
\bibinfo{author}{A.~Artale}, \bibinfo{author}{J.~C. Jung}, \bibinfo{author}{A.~Mazzullo}, \bibinfo{author}{A.~Ozaki}, \bibinfo{author}{F.~Wolter},
\newblock \bibinfo{title}{Living without beth and craig: Definitions and interpolants in description and modal logics with nominals and role inclusions},
\newblock \bibinfo{journal}{{ACM} Trans. Comput. Log.} \bibinfo{volume}{24} (\bibinfo{year}{2023}) \bibinfo{pages}{34:1--34:51}. \URLprefix \url{https://doi.org/10.1145/3597301}. \DOIprefix\doi{10.1145/3597301}.
\bibitem[{Kurucz et~al.(2023)Kurucz, Wolter, and Zakharyaschev}]{DBLP:conf/kr/KuruczWZ23}
\bibinfo{author}{A.~Kurucz}, \bibinfo{author}{F.~Wolter}, \bibinfo{author}{M.~Zakharyaschev},
\newblock \bibinfo{title}{Definitions and (uniform) interpolants in first-order modal logic},
\newblock in: \bibinfo{booktitle}{Proceedings of the 20th International Conference on Principles of Knowledge Representation and Reasoning, {KR} 2023, Rhodes, Greece, September 2-8, 2023}, \bibinfo{year}{2023}, pp. \bibinfo{pages}{417--428}. \URLprefix \url{https://doi.org/10.24963/kr.2023/41}. \DOIprefix\doi{10.24963/KR.2023/41}.
\bibitem[{Konev et~al.(2009)Konev, Lutz, Walther, and Wolter}]{KonevLWW09}
\bibinfo{author}{B.~Konev}, \bibinfo{author}{C.~Lutz}, \bibinfo{author}{D.~Walther}, \bibinfo{author}{F.~Wolter},
\newblock \bibinfo{title}{Formal properties of modularisation},
\newblock in: \bibinfo{editor}{H.~Stuckenschmidt}, \bibinfo{editor}{C.~Parent}, \bibinfo{editor}{S.~Spaccapietra} (Eds.), \bibinfo{booktitle}{Modular Ontologies: Concepts, Theories and Techniques for Knowledge Modularization}, volume \bibinfo{volume}{5445} of \textit{\bibinfo{series}{Lecture Notes in Computer Science}}, \bibinfo{publisher}{Springer}, \bibinfo{year}{2009}, pp. \bibinfo{pages}{25--66}. \URLprefix \url{https://doi.org/10.1007/978-3-642-01907-4\_3}. \DOIprefix\doi{10.1007/978-3-642-01907-4\_3}.
\bibitem[{Kuijer et~al.(2025)Kuijer, Tan, Wolter, and Zakharyaschev}]{kuijer2025separationdefinabilityfragmentstwovariable}
\bibinfo{author}{L.~Kuijer}, \bibinfo{author}{T.~Tan}, \bibinfo{author}{F.~Wolter}, \bibinfo{author}{M.~Zakharyaschev}, \bibinfo{title}{Separation and definability in fragments of two-variable first-order logic with counting}, \bibinfo{year}{2025}. \URLprefix \url{https://arxiv.org/abs/2504.20491}. \href{http://arxiv.org/abs/2504.20491}{{\tt arXiv:2504.20491}}.
\bibitem[{Baader et~al.(2017)Baader, Horrocks, Lutz, and Sattler}]{DL-Textbook}
\bibinfo{author}{F.~Baader}, \bibinfo{author}{I.~Horrocks}, \bibinfo{author}{C.~Lutz}, \bibinfo{author}{U.~Sattler}, \bibinfo{title}{An Introduction to Description Logic}, \bibinfo{publisher}{Cambridge University Press}, \bibinfo{year}{2017}.
\bibitem[{Jung and Kolodziejski(2025)}]{DBLP:conf/stacs/JungK25}
\bibinfo{author}{J.~C. Jung}, \bibinfo{author}{J.~Kolodziejski},
\newblock \bibinfo{title}{Modal separation of fixpoint formulae},
\newblock in: \bibinfo{editor}{O.~Beyersdorff}, \bibinfo{editor}{M.~Pilipczuk}, \bibinfo{editor}{E.~Pimentel}, \bibinfo{editor}{K.~T. Nguyen} (Eds.), \bibinfo{booktitle}{42nd International Symposium on Theoretical Aspects of Computer Science, {STACS} 2025, March 4-7, 2025, Jena, Germany}, volume \bibinfo{volume}{327} of \textit{\bibinfo{series}{LIPIcs}}, \bibinfo{publisher}{Schloss Dagstuhl - Leibniz-Zentrum f{\"{u}}r Informatik}, \bibinfo{year}{2025}, pp. \bibinfo{pages}{55:1--55:20}. \DOIprefix\doi{10.4230/LIPICS.STACS.2025.55}.
\bibitem[{Koopmann and Schmidt(2014)}]{Koopmann2014CountAF}
\bibinfo{author}{P.~Koopmann}, \bibinfo{author}{R.~Schmidt},
\newblock \bibinfo{title}{Count and forget: Uniform interpolation of shq-ontologies},
\newblock in: \bibinfo{booktitle}{International Joint Conference on Automated Reasoning}, \bibinfo{year}{2014}. \URLprefix \url{https://api.semanticscholar.org/CorpusID:38753229}.

\end{thebibliography}

\newpage

\appendix
\section{Missing Proof Details}

\lemalcjointconsistency*

\begin{proof}[Proof (Sketch)]
For implication ``$\Rightarrow$'', suppose that $C_0,D_0$ are jointly $\sim_{\ALC,\Sigma}$-consistent under \Omc, that is, there are models $\Imc_1,\Imc_2$ of \Omc and elements $d_1\in C_0^{\Imc_1}$ and $d_2\in D_0^{\Imc_2}$ such that $\Imc_1,d_1\sim_{\ALC,\Sigma} \Imc_2,d_2$. Since \ALCH is preserved under taking disjoint unions, we can assume without loss of generality that $\Imc_1=\Imc_2=\Imc$. We read off a set $\Mmc^*$ of mosaics by taking
\[\Mmc^* = \{\ \{\tp_\Imc(e)\mid \Imc,d\sim_{\ALC,\Sigma} \Imc,e\}\ \mid d\in \Delta^\Imc\}.\]
It is routine to verify that $\Mmc^*$ satisfies the conditions formulated in Lemma~\ref{lem:alc-joint-consistency}. 

For implication ``$\Leftarrow$'', let $\Mmc^*$ be a set of mosaics that does not contain a bad mosaic and such that there is $T^*\in\Mmc^*$ and $t_1,t_2\in T^*$ with $C\in t_1$ and $D\in t_2$. We construct an interpretation \Imc as follows: 
\begin{align*}
    \Delta^{\Imc} & = \{(t,T)\mid T\in \Mmc^*\text{ and }t\in T\} \\
    A^{\Imc} & = \{(t,T)\in \Delta^\Imc\mid A\in t\} \\
    r^{\Imc} & = \{( (t,T),(t',T') )\in \Delta^\Imc\times \Delta^\Imc\mid t\leadsto_r t'\text{ and for all $s\in \Sigma$: }((\Omc\models r\sqsubseteq s) \Rightarrow T\leadsto_sT')\}
\end{align*}
One can verify by structural induction that $C\in (t,T)$ iff $(t,T)\in C^\Imc$, for all $C\in \sub(\Omc,C_0,D_0)$ and $(t,T)\in \Delta^\Imc$. Consequently, $(t_1,T^*)\in C_0^\Imc$ and $(t_2,T^*)\in D_0^\Imc$. Moreover, following relation $Z$:
\[Z=\{(t,T),(t',T)\mid T\in \Mmc^*\}\]
is a $\Sigma$-bisimulation. Since $((t_1,T^*), (t_2,T^*))\in Z$, we conclude that $C_0,D_0$ are jointly $\sim_{\ALC,\Sigma}$-consistent under \Omc.
\end{proof}
\end{document}